\documentclass[journal]{IEEEtran}
\IEEEoverridecommandlockouts
\usepackage{amsthm}
\usepackage{amsmath}
\usepackage{mathtools}
\usepackage{amssymb}
\usepackage{algorithmic}
\usepackage{algorithm}
\usepackage{verbatim}
\usepackage{cite}
\usepackage{url}
\usepackage{cases}
\usepackage{subfigure}

\newtheorem{proposition}{Proposition}
\newtheorem{theorem}{Theorem}

\begin{document}
\title{Sub-Nyquist Sampling for Power Spectrum Sensing in Cognitive Radios: A Unified Approach} 
\author{Deborah Cohen, \emph{Student IEEE}
        and Yonina C. Eldar, \emph{Fellow IEEE}}
\maketitle

\begin{abstract}
In light of the ever-increasing demand for new spectral bands and the underutilization of those already allocated, the concept of Cognitive Radio (CR) has emerged. Opportunistic users could exploit temporarily vacant bands after detecting the absence of activity of their owners. One of the crucial tasks in the CR cycle is therefore spectrum sensing and detection which has to be precise and efficient. Yet, CRs typically deal with wideband signals whose Nyquist rates are very high. In this paper, we propose to reconstruct the power spectrum of such signals from sub-Nyquist samples, rather than the signal itself as done in previous work, in order to perform detection. We consider both sparse and non sparse signals as well as blind and non blind detection in the sparse case. For each one of those scenarii, we derive the minimal sampling rate allowing perfect reconstruction of the signal's power spectrum in a noise-free environment and provide power spectrum recovery techniques that achieve those rates. The analysis is performed for two different signal models considered in the literature, which we refer to as the analog and digital models, and shows that both lead to similar results. Simulations demonstrate power spectrum recovery at the minimal rate in noise-free settings and show the impact of several parameters on the detector performance, including signal-to-noise ratio (SNR), sensing time and sampling rate.
\end{abstract}

\IEEEpeerreviewmaketitle

\section{Introduction}
Spectral resources are traditionally allocated to licensed or primary users (PUs) by governmental organizations. Today, most of the spectrum is already owned and new users can hardly find free frequency bands. In light of the ever-increasing demand from new wireless communication users, this issue has become critical over the past few years. On the other hand, various studies \cite{Study1, Study2, study3} have shown that this over-crowded spectrum is usually significantly underutilized and can be described as the union of a small number of narrowband transmissions spread across a wide spectrum range. This is the motivation behind cognitive radio (CR), which would allow secondary users to opportunistically use the licensed spectrum when the corresponding PU is not active \cite{Mitola, Haykin}. Even though the concept of CR is said to have been introduced by Mitola \cite{Mitola, MitolaMag}, the idea of learning machines for spectrum sensing can be traced back to Shannon \cite{Shannon}.

One of the crucial tasks in the CR cycle is spectrum sensing \cite{cog}. The CR has to constantly monitor the spectrum and detect the PU's activity in order to select unoccupied bands, before and throughout its transmission. At the receiver, the CR samples the signal and performs detection to assert which band is unoccupied and can be exploited for opportunistic transmissions. In order to minimize the interference that could be caused to PUs, the spectrum sensing task performed by a CR should be reliable and fast \cite{cognitive1, cognitive2, WidebandMishali}. On the other hand, in order to increase the chance to find an unoccupied spectral band, the CR has to sense a wide band of spectrum. Nyquist rates of wideband signals are high and can even exceed today's best analog-to-digital converters (ADCs) front-end bandwidths. Besides, such high sampling rates generate a large number of samples to process, affecting speed and power consumption. 

To overcome the rate bottleneck, several new sampling methods have recently been proposed \cite{Mishali_theory, Mishali_multicoset, MagazineMishali} that reduce the sampling rate in multiband settings below the Nyquist rate. In \cite{Mishali_theory, Mishali_multicoset, MagazineMishali}, the authors derive the minimal sampling rate allowing for perfect signal reconstruction in noise-free settings and provide sampling and recovery techniques. However, when the final goal is spectrum sensing and detection, reconstructing the original signal is unnecessary. Following the ideas in \cite{Leus, Leus2, Davies, Davies2}, we propose, in this paper, to only reconstruct the signal's power spectrum from sub-Nyquist samples, in order to perform signal detection.

Several papers have considered power spectrum reconstruction from sub-Nyquist samples, by treating two different signal models.
The first, and most popular so it seems, is a digital model which is based upon a linear relation between the sub-Nyquist and Nyquist samples obtained for a given sensing time frame. Ariananda et al. \cite{Leus, Leus2} have deeply investigated this model with multicoset sampling \cite{Mishali_multicoset, Bresler}. They consider both time and frequency domain approaches and discuss the reconstruction of the autocorrelation or power spectrum respectively, from undertermined and overdermined systems. For the first case, they expoit sparsity properties of the signal and apply compressed sensing (CS) reconstruction techniques but do not analyze the sampling rate. The authors rather focus the analysis on the second case, namely the overdetermined system, and show that it can be solved without any sparsity assumption. They demonstrate that the so-called minimal sparse ruler patterns \cite{ruler} provide a sub-optimal solution for sub-Nyquist sampling, when using multicoset sampling.

The second is an analog model that treats the class of wide-sense stationary multiband signals, whose frequency support lies within several continuous intervals (bands). Here, a linear relation between the Fourier transform of the sub-Nyquist samples and frequency slices of the original signal's spectrum is exploited.
In \cite{Davies, Davies2}, the authors propose a method to estimate finite resolution approximations to the true power spectrum exploiting multicoset sampling. That is, they estimate the average power within subbands rather than power spectrum for each frequency. They consider overdetermined and undertermined, or compressive systems. In the latter case, CS techniques are used, which exploit the signal's sparsity, whereas the former setting does not assume any sparsity. In \cite{Davies}, the authors assume that the sampling pattern is such that the system they obtain has a unique solution but no specific sampling pattern or rate satisfying this condition is discussed. In \cite{Davies2}, sampling patterns generated uniformly at random and the Golomb ruler are considered in simulations but no analysis of the required rate is performed. Another recent paper \cite{wang} considers the analog model with multicoset samplingin the non sparse setting. The authors derive necessary and sufficient conditions for perfect power spectrum reconstruction in noise free settings. They show that any universal sampling pattern guarantees perfect recovery under that sufficient conditions. They further investigate two other sub-optimal patterns that lead to perfect reconstruction under lower sampling rates.

In this paper, we aim at filling several gaps in the current literature. First, to the best of our knowledge, no comparison has been made between the two models and their respective results. Second, the general conditions required from the sampling matrix and the resulting minimal sampling rate for perfect power spectrum reconstruction in a noiseless environment have not been analyzed. In \cite{Leus, Leus2}, only multicoset sampling is considered and no universal minimal rate is provided. Rather, several compression ratios given by the sub-optimal solution of the minimal sparse ruler are shown to suffice. In \cite{Davies, Davies2}, no proof of the uniqueness of the solution is given. The authors in \cite{wang} provide necessary and sufficient conditions for perfect recovery, but only for the analog model in the non sparse setting. In this paper, we aim at providing a unifying framework for power spectrum reconstruction from sub-Nyquist samples by bridging between the two models.

We thus consider the two different signal models: the analog or multiband model and the digital one that we relate to the multi-tone model in order to anchor it to the original analog signal. For the analog model, we focus on sampling schemes that operate on the bins of the signal's spectrum and provide samples that are linear transformations of these. Two examples of such schemes are the sampling methods proposed in \cite{Mishali_theory, Mishali_multicoset, MagazineMishali}, namely multicoset sampling and the Modulated Wideband Converter (MWC). For the digital model, we analyse a generic sampling scheme and provide two different reconstruction approaches. The first, considered for example in \cite{Leus, Leus2}, is performed in the time domain whereas the second is realized in the frequency domain. While the analysis of the conditions for perfect reconstruction turns out to be difficult in the time domain, we show that it is convenient in the frequency one. There, both the analog and the digital model lead to similar relations and can therefore be investigated jointly. It is interesting to notice that other applications based on sub-Nyquist sampling, such as radar \cite{radar}, use frequency domain analysis as well. 
We examine three different scenarii: (1) the signal is not assumed to be sparse, (2) the signal is assumed to be sparse and the carrier frequencies of the narrowband transmissions are known, (3) the signal is sparse but we do not assume carrier knowledge. 

The main contributions of this paper are twofold. First, for each one of the scenarii, we derive the minimal sampling rate for perfect power spectrum reconstruction with respect to our settings in a noise-free environment. We show that the rate required for power spectrum reconstruction is half the rate that allows for perfect signal reconstruction, for each one of the scenarii, namely the Nyquist rate, the Landau rate \cite{LandauCS} and twice the Landau rate \cite{Mishali_multicoset}. Second, we present reconstruction techniques that achieve those rates for both signal models. Throughout the paper, minimal sampling rate refers to the lowest rate enabling perfect reconstruction of the power spectrum in a noiseless environment for a general sampling scheme. We do not consider the minimal rate achievable for a specific design of the sampling system. For instance, in \cite{Leus, Leus2}, the authors show that designing the multicoset sampling matrix according to the minimal sparse ruler pattern results in a minimal rate below ours. Some other specific sampling patterns are considered in \cite{wang}. In contrast, we focus on generic systems without any particular structure.

This paper is organized as follows. In Section \ref{ModelProb}, we present the stationary multiband and multi-tone models and formulate the problem. Section \ref{SecOpt} describes the sub-Nyquist sampling stage and ties the original signal's power spectrum to correlation between the samples. In Section \ref{sec:rate}, we derive the minimal sampling rate for each one of the three scenarii described above and present recovery techniques that achieve those rates. Numerical experiments are presented in Section \ref{sec:simulations}. We demonstrate power spectrum reconstruction from sub-Nyquist samples, show the impact of several practical parameters on the detection performance, and compare our detection results to Nyquist rate sampling and to spectrum based detection from sub-Nyquist samples \cite{Mishali_theory}.

\section{System Models and Goal}
\label{ModelProb}

\subsection{Analog Model}
\label{sec:model1}

Let $x(t)$ be a real-valued continuous-time signal, supported on $\mathcal{F} = [-T_{\text{Nyq}}/2, +T_{\text{Nyq}}/2]$ and composed of up to $N_{\text{sig}}$ uncorrelated stationary transmissions, such that 
\begin{equation}
x(t)=\sum_{i=1}^{N_{\text{sig}}} \rho_i s_i(t).
\end{equation}
Here $\rho_i \in \{0,1\}$ and $s_i(t)$ is a zero-mean wide-sense stationary signal. The value of $\rho_i$ determines whether or not the $i$th transmission is active. The bandwidth of each transmission is assumed to not exceed $2B$ (where we consider both positive and negative frequency bands). Formally, the Fourier transform of $x(t)$ defined by
\begin{equation}
X(f)=\int_{-\infty}^{\infty}x(t)e^{-j2\pi f t} \mathrm{d} t
\end{equation}
is zero for every $f \notin \mathcal{F}$. We denote by $f_{\text{Nyq}} = 1/T_{\text{Nyq}}$ the Nyquist rate of $x(t)$ and by $S_x$ the support of $X(f)$.

The power spectrum of $x(t)$ is the Fourier transform of its autocorrelation, namely
\begin{equation}
\label{eq:spec}
P_x(f)=\int_{-\infty}^{\infty}r_x(\tau)e^{-j2\pi f \tau} \mathrm{d} \tau,
\end{equation}
where $r_x(\tau) = \mathbb{E} \left[ x(t)x(t-\tau) \right]$ is the autocorrelation function of $x(t)$. 
From \cite{Papoulis}, it holds that
\begin{equation}
P_x(f)=\mathbb{E} \left| X(f) \right|^2.
\end{equation}
Thus, obviously, the support of $P_x(f)$  is identical to that of $X(f)$, namely $S_x$. Our goal is to reconstruct $P_x(f)$ from sub-Nyquist samples. In Section \ref{SecOpt}, we describe our sampling schemes and show how one can relate $P_x(f)$ to correlation of the samples.

We consider three different scenarii.

\subsubsection{No sparsity assumption}
In the first scenario, we assume no \emph{a priori} knowledge on the signal and we do not suppose that $x(t)$ is sparse, namely $N_{\text{sig}}B$ can be on the order of $f_{\text{Nyq}}$.

\subsubsection{Sparsity assumption and non blind detection}
Here, we assume that $x(t)$ is sparse, namely $N_{\text{sig}}B \ll f_{\text{Nyq}}$. We denote $K_f=2N_{\text{sig}}$. Moreover, the support of the potentially active transmissions is known and corresponds to the frequency support of licensed users defined by the communication standard. However, since the PUs' activity can vary over time, we wish to develop a detection algorithm that is independent of a specific known signal support.

\subsubsection{Sparsity assumption and blind detection}
In the last scenario as in the previous one, we assume that $x(t)$ is sparse, but we do not assume any \emph{a priori} knowledge on the carrier frequencies. Only the maximal number of transmissions $N_{\text{sig}}$ and the maximal bandwidth $2B$ are assumed to be known.

\subsection{Digital Model}
\label{sec:model2}

The second model we consider is the multi-tone model. Let $x(t)$ be a continuous-time signal defined over the interval $[0,T)$ and composed of up to $N_{\text{sig}}$ transmissions, such that
\begin{equation}
x(t)=\sum_{i=1}^{N_{\text{sig}}} \rho_i s_i(t), \qquad t \in [0,T).
\end{equation}
Again, $\rho_i \in \{0,1\}$ and $s_i(t)$ is a wide-sense stationary signal.
Since $x(t)$ is defined over $[0,T)$, it has a Fourier series representation
\begin{equation}
x(t) = \sum_{k = -Q/2}^{Q/2} c[k] e^{j \frac{2 \pi k}{T}t}, \qquad t \in [0,T),
\label{xmodel}
\end{equation}
where $Q/(2T)$ is the maximal possible frequency in $x(t)$.
Each transmission $s_i(t)$ has a finite number of Fourier coefficients, up to $2K_{max} \le Q+1$, so that
\begin{equation}
s_i(t) = \sum_{k \in \Omega_i} c[k] e^{j \frac{2 \pi k}{T}t}, \qquad t \in [0,T),
\label{smodel}
\end{equation}
where $\Omega_i$ is a set of integers with $\left| \Omega_i \right| \le 2K_{\text{max}}$ and $\max_{k \in \{\Omega_i\}} |k|\le Q/2$. Thus, here the support $S_x$ of $x(t)$ is $S_x=\bigcup_{i=1}^{N_{\text{sig}}} \Omega_i$.

For mathematical convenience, for this model we will consider the Nyquist samples of $x(t)$, namely
\begin{equation}
\label{xsamp}
x[n]=x(n T_{\text{Nyq}}), \qquad 0 \le n < T/T_{\text{Nyq}},
\end{equation}
where $T_{\text{Nyq}}=T/(Q+1)$. Since $x(t)$ is wide-sense stationary, it follows that $\mathbf{x}$ is wide-sense stationary as well. Let us define $N=T/T_{\text{Nyq}}= Q+1$.
From (\ref{xmodel}), the autocorrelation of $\mathbf{x}$, namely $r_\mathbf{x}[\nu] = \mathbb{E} \left[ x[n]x[n-\nu] \right]$, has a Fourier representation
\begin{equation}
r_\mathbf{x}[\nu] = \sum_{k = -Q/2}^{Q/2} s_{\mathbf{x}}[k] e^{j \frac{2 \pi k}{N} \nu}, \qquad 0 \le \nu \le N-1,
\label{eq:spec_autoco}
\end{equation}
where
\begin{equation}
s_{\mathbf{x}}[k]=\mathbb{E} \left[ c^2[k] \right] , \qquad -\frac{Q}{2} \le k \le \frac{Q}{2}.
\label{eq:spec2}
\end{equation}
From the stationarity property of the signal, namely $r_\mathbf{x}[\nu]$ is a function of $\nu$ only, it holds that
\begin{equation}
\mathbb{E} \left[ c[k] c^*[l] \right] =0, \qquad -\frac{Q}{2} \le k \neq l \le \frac{Q}{2}.
\label{eq:spec3}
\end{equation}
From (\ref{eq:spec2}), it is obvious that the Fourier coefficients of $r_x[\nu]$ lie in the same support as those of $x(t)$, namely $S_x$.

Again, we consider three different scenarii.

\subsubsection{No sparsity assumption}
In the first scenario, we assume no \emph{a priori} knowledge on the signal and we do not suppose that $x(t)$ is sparse, namely $N_{\text{sig}}K_{\text{max}}$ can be on the order of $Q+1$.

\subsubsection{Sparsity assumption and non blind detection}
Here, we assume that $x(t)$ is sparse, namely $N_{\text{sig}}K_{\text{max}} \ll Q+1$ and that the Fourier frequencies in the Fourier series expansions of $s_i(t)$, namely $\Omega_i, 1 \leq i \leq N_{\text{sig}}$ are known. We denote $K_f=2N_{\text{sig}}K_{\text{max}}$.

\subsubsection{Sparsity assumption and blind detection}
In the last scenario, we assume that $x(t)$ is sparse but we do not assume any \emph{a priori} knowledge on the Fourier frequencies in the Fourier series expansions of $s_i(t)$.

\subsection{Problem Formulation}

In each one of the scenarii defined in the previous section, our goal is to assess which of the $N_{\text{sig}}$ transmissions are active from sub-Nyquist samples of $x(t)$. For each signal, we define the hypothesis $\mathcal{H}_{i,0}$ and $\mathcal{H}_{i,1}$, namely the $i${th} transmission is absent and active, respectively. 

In order to determine which of the $N_{\text{sig}}$ transmissions are active, we first reconstruct the power spectrum of $x(t)$ for the first model (\ref{eq:spec}), or the Fourier coefficients of the signal's sampled autocorrelation for the second one (\ref{eq:spec2}). In the first and third scenarii, we fully reconstruct the power spectrum. In the second one, we exploit our prior knowledge and reconstruct it only at the potentially occupied locations. We can then perform detection on the fully or partially reconstructed power spectrum. Note that, to do so, we do not sample $x(t)$ at its Nyquist rate, nor compute its Nyquist rate samples. For each one of the scenarii, we derive the minimal rate enabling perfect reconstruction of (\ref{eq:spec}) and (\ref{eq:spec2}) respectively, in a noise-free environment, and present recovery techniques that achieve those rates.

By performing e.g. energy detection on the reconstructed power spectrum, we can detect unoccupied spectral bands, namely spectrum holes, from sub-Nyquist samples. This makes the detection process faster, more efficient and less power consuming, which fits the requirements of CRs. Other forms of detection are also possible, once the power spectrum is recovered.

\section{Spectrum Reconstruction from sub-Nyquist Samples}
\label{SecOpt}

\subsection{Analog Model: Sampling and the Analog Spectrum}

We begin with the analog model. For this model, we consider two different sampling schemes: multicoset sampling \cite{Mishali_multicoset} and the MWC \cite{Mishali_theory} which were previously proposed for sparse multiband signals. We show that both schemes lead to identical expressions of the signal's power spectrum in terms of that of the samples. In this section, we consider reconstruction of the whole power spectrum. In Section \ref{rate2}, we show how we can reconstruct the power spectrum only at potentially occupied locations when we have \emph{a priori} knowledge on the carrier frequencies.

\subsubsection{Multicoset sampling}
\label{sec:multico}
Multicoset sampling \cite{Bresler} can be described as the selection of certain samples from the uniform grid. More precisely, the uniform grid is divided into blocks of $N$ consecutive samples, from which only $M$ are kept. The $i$th sampling sequence is defined as
\begin{equation}
x_{c_i}[n]= \left\{ \begin{array}{ll}
x(nT_{\text{Nyq}}), & n=mN+c_i, m \in \mathbb{Z} \\
0, & \text{otherwise},
\end{array} \right.
\end{equation}
where $0 < c_1 < c_2 < \dots < c_M < N-1$. Let $f_s = \frac{1}{NT_{\text{Nyq}}} \ge B$ be the sampling rate of each channel and $\mathcal{F}_s=[-f_s/2, f_s/2]$.
Following the derivations from multicoset sampling \cite{Mishali_multicoset}, we obtain
\begin{equation}
\mathbf{z}(f) = \mathbf{A} \mathbf{x}(f), \qquad f \in \mathcal{F}_s,
\label{eq:multico}
\end{equation}
where $\mathbf{z}_i(f) = X_{c_i}(e^{j2\pi f T_{\text{Nyq}}}), 0 \le i \le M-1$ are the discrete-time Fourier transforms (DTFTs) of the multicoset samples and
\begin{equation}
\mathbf{x}_k(f)=X\left(f+K_kf_s \right), \quad 1 \le k \le N,
\label{xdef}
\end{equation}
where $K_k = k-\frac{N+1}{2}, 1 \le k \le N$ for odd $N$ and $K_k = k-\frac{N+2}{2}, 1 \le k \le N$ for even $N$. Each entry of $\mathbf{x}(f)$ is referred to as a bin since it consists of a slice of the spectrum of $x(t)$.
The $ik$th element of the $M \times N$ matrix $\mathbf{A}$ is given by
\begin{equation}
\mathbf{A}_{ik} = \frac{1}{NT_{\text{Nyq}}} e^{j\frac{2 \pi}{N} c_i K_k}.
\end{equation}

\subsubsection{MWC sampling}
The MWC \cite{Mishali_theory} is composed of $M$ parallel channels. In each channel, an analog mixing front-end, where $x(t)$ is multiplied by a mixing function $p_i(t)$, aliases the spectrum, such that each band appears in baseband. The mixing functions $p_i(t)$ are required to be periodic. We denote by $T_p$ their period and we require $f_p=1/T_p \ge B$.
The function $p_i(t)$  has a Fourier expansion
\begin{equation}
p_i(t) =\sum_{l=-\infty}^{\infty} c_{il} e^{j\frac{2\pi}{T_p} lt}.
\end{equation} 
In each channel, the signal goes through a lowpass filter with cut-off frequency $f_s/2$ and is sampled at rate $f_s \ge f_p $. For the sake of simplicity, we choose $f_s=f_p$. The overall sampling rate is $Mf_s$ where $M \le N=f_{\text{Nyq}}/f_s$.
Repeating the calculations in \cite{Mishali_theory}, we derive the relation between the known DTFTs of the samples $z_i[n]$ and the unknown $X(f)$
\begin{equation}
\mathbf{z}(f)=\mathbf{A}\mathbf{x}(f), \qquad f \in \mathcal{F}_s,
\label{eq:mwc}
\end{equation}
where $\mathbf{z}(f)$ is a vector of length $M$ with $i$th element $\mathbf{z}_i(f)=Z_i(e^{j2\pi fT_s})$. The unknown vector $\mathbf{x}(f)$ is given by (\ref{xdef}). The $M \times N$ matrix $\mathbf{A}$ contains the coefficients $c_{il}$:
\begin{equation}
\mathbf{A}_{il} = c_{i,-l}=c^*_{il}.
\end{equation}
For both sampling schemes, the overall sampling rate is
\begin{equation}
f_{tot}=Mf_s=\frac{M}{N}f_{\text{Nyq}}.
\end{equation}

\subsubsection{Analog Power Spectrum Reconstruction}
\label{analog_rec}
We note that systems (\ref{eq:multico}) and (\ref{eq:mwc}) are identical for both sampling schemes. The only difference is the sampling matrix $\mathbf{A}$. We assume that $\bf{A}$ is full spark in both cases \cite{Mishali_multicoset, Mishali_theory}, namely, that every $M$ columns of $\bf A$ are linearly independent. We thus can derive a method for reconstruction of the analog power spectrum for both sampling schemes together. We will reconstruct $P_x(f)$ from the correlation between $\mathbf{z}(f)$, defined in (\ref{eq:multico}) and (\ref{eq:mwc}).

Since $x(t)$ is a wide-sense stationary process, we have \cite{Papoulis}
\begin{equation}
\label{eq:papou}
\mathbb{E} [X(f_1) X^*(f_2) ] = P_x(f_1) \delta (f_1-f_2)
\end{equation}
where $P_x(f)$ is given by (\ref{eq:spec}). We define the autocorrelation matrix $\mathbf {R_x}(f) = \mathbb{E} [\mathbf{x}(f) \mathbf{x}^H(f) ]$, where  $(.)^H$ denotes the Hermitian operation. From (\ref{eq:papou}),  $\mathbf{R_x}(f)$ is a diagonal matrix with $\mathbf{R}_{\mathbf{x}_{(i,i)}}(f)=P_x(f+ K_i f_s)$ \cite{Davies}, where $K_i$ is defined in Section \ref{sec:multico}. Clearly, our goal can be stated as recovery of $\mathbf{R_x}(f)$, since once $\mathbf{R_x}(f)$ is known, $P_x(f)$ follows for all $f$.

We now relate $\mathbf{R_x}(f)$ to the correlation of the sub-Nyquist samples.
From (\ref{eq:multico}) or (\ref{eq:mwc}), we have
\begin{equation}
\mathbf{R_z}(f) = \mathbf{A} \mathbf{R_x}(f) \mathbf{A}^H, \qquad f \in \mathcal{F}_s,
\label{eq:autoco2}
\end{equation}
where $\mathbf {R_z}(f) = \mathbb{E} [\mathbf{z}(f) \mathbf{z}^H(f) ]$. It follows that 
\begin{equation}
\mathbf{r_z}(f) = \mathbf{(\bar{A} \otimes A)}\text{vec}(\mathbf{R}_\mathbf{x}(f) =  \mathbf{(\bar{A} \otimes A)} \mathbf{B} \mathbf{r_x}(f) \triangleq \mathbf{\Phi} \mathbf{r_x}(f),
\label{eq:rzrx}
\end{equation}
where $\bf \Phi=(\bar{A} \otimes A)B= \bar{A} \odot A$, and $\bf \bar{A}$ denotes the conjugate matrix of $\bf A$. 
Here $\otimes$ is the Kronecker product, $\odot$ denotes the Khatri-Rao product, $\mathbf{r_z}(f) = \text{vec}(\mathbf{R_z}(f))$, and $\bf B$ is a $N^2 \times N$ selection matrix that has a 1 in the $j$th column and $[(j-1)N+j]$th row, $1 \le j \le N$ and zeros elsewhere. Thus, $\mathbf{r}_{\mathbf{x}_i}(f)=P_x(f+ K_i f_s)$ and by recovering $\mathbf{r_x}(f)$ for all $f \in \mathcal{F}_s$, we recover the entire power spectrum of $x(t)$. 

We now discuss the sparsity of $\mathbf{r_x}(f)$ for the second and third scenarii. We chose $f_s \ge B$ so that each transmission contributes only a single non zero element to $\mathbf{r_x}(f)$ (referring to a specific $f$), and consequently $\mathbf{r_x}(f)$ has at most $K_f \ll N$ non zeros for each $f$ \cite{Mishali_theory}, corresponding to $S_x$.
In the next section, we derive conditions on the sampling rate for (\ref{eq:rzrx}) to have a unique solution.

It is interesting to note that (\ref{eq:rzrx}), which is written in the frequency domain, is valid in the time domain as well. We can therefore estimate $\mathbf{r_z}(f)$ and reconstruct $\mathbf{r_x}(f)$ in the frequency domain, or alternatively, we can estimate $\mathbf{r_z}[n]$ and reconstruct $\mathbf{r_x}[n]$ in the time domain using
\begin{equation}
\mathbf{r_z}[n] = \mathbf{\Phi} \mathbf{r_x}[n].
\label{eq:time}
\end{equation}
Note that $\mathbf{r_x}(f)$ is $K_f$-sparse for each specific frequency $f \in \mathcal{F}_S$, whereas $\mathbf{r_x}[n]$ is $2K_f$-sparse since each transmission can be split into two bins. Therefore, in Section \ref{sec:scen3}, we show that the minimal sampling rate is achieved only in the frequency domain.
Since the vectors $\mathbf{r_x}[n]$ are jointly sparse, we can recover the support $S_x$ from one sample in each channel, provided that the value of the samples in the occupied bins is not zero for each $n$. However, in order to ensure robustness to noise and better performance, we consider more than one sample in the simulations.

As a final comment, below we assume full knowledge of $\mathbf{r_z}(f)$ or $\mathbf{r_z}[n]$, or the possibility to compute them. In Section \ref{sec:simulations}, we show how to approximate $\mathbf{r_z}(f)$ and $\mathbf{r_z}[n]$ from a finite data block.

\subsection{Discrete Model: Reconstruction of the Digital Spectrum}
\label{Discrete}
In this model, we wish to recover the Fourier coefficients of the autocorrelation of $\mathbf{x}$, defined in (\ref{eq:spec2}). The traditional approach in this setting exploits the time domain characteristics of the stationary signal. Unfortunately, the analysis of the recovery conditions of the Fourier coefficients of $\mathbf{x}$ turns out to be quite involved. Therefore, we propose a second approach, that exploits the equivalent frequency domain properties of the signal. We show that in that case, the same analysis as for the analog model can be performed.

\subsubsection{Time domain}

Define the autocorrelation matrix as
\begin{eqnarray}
\mathbf{R_x}&=& \mathbb{E} \left[ \mathbf{x}[n]{\mathbf{x}}^H [n-\nu] \right] \\
&=& \left[ \begin{array}{cccc}
r_\mathbf{x}[0] & r_\mathbf{x}[1] & \dots & r_\mathbf{x}[N-1] \\
r_\mathbf{x}[1] & r_\mathbf{x}[0] & \dots & r_\mathbf{x}[N-2] \\
\vdots & \vdots & \ddots & \vdots \\
r_\mathbf{x}[N-1] & r_\mathbf{x}[N-2] & \dots &r_\mathbf{x}[0]
\end{array} \right]. \nonumber
\end{eqnarray}
From (\ref{eq:spec_autoco}),
\begin{equation}
\mathbf{s}_\mathbf{x} = \mathbf{F} \mathbf{r}_\mathbf{x},
\end{equation}
where $\bf s_x$ is defined in (\ref{eq:spec2}), $\mathbf{F}$ is the $N \times N$ DFT matrix and
\begin{equation}
\mathbf{r_x}= \left[ \begin{array}{cccc}
r_\mathbf{x}[0]  &
r_\mathbf{x}[1]  &
\dots  &
r_\mathbf{x}[N-1] 
\end{array} \right]^T.
\end{equation}
Therefore,
\begin{equation} \label{eq:vecx}
\text{vec}(\mathbf{R}_\mathbf{x}) = \mathbf{\tilde{B}} \mathbf{r_x} =\frac{1}{N}\mathbf{\tilde{B}} \mathbf{F}^{H} \mathbf{s}_\mathbf{x},
\end{equation}
where $\mathbf{\tilde{B}}$ is a $N^2 \times N$ repetition matrix whose $i$th row is given by the $\left[ \left| \lfloor \frac{i-1}{N} \rfloor - (i-1) \mod N \right| +1\right] $th row of the $N \times N$ identity matrix.

We now relate $\textbf{s}_\mathbf{x}$ to the covariance matrix of the sub-Nyquist samples ${\mathbf{R}_\mathbf{z}} = \mathbb{E} \left[ \mathbf{z}\mathbf{z}^{H} \right]$.
We start by deriving the relationship between $\textbf{R}_\mathbf{z}$ and $\textbf{R}_\mathbf{x}$. From (\ref{eq:dmodel1}), we have
\begin{equation}
\bf R_z = AR_xA^H.
\label{eq:autoco}
\end{equation}
Here, we assume that $\bf A$ is full spark.
Vectorizing both sides of (\ref{eq:autoco}) and using (\ref{eq:vecx}), we obtain
\begin{equation}
\mathbf{r_z} = \mathbf{(\bar{A} \otimes A)}\text{vec}(\mathbf{R}_\mathbf{x}) =  \frac{1}{N} \mathbf{(\bar{A} \otimes A)} \mathbf{\tilde{B}} \mathbf{F}^{H} \mathbf{s_x} \triangleq \mathbf{\Psi s_x},
\label{eq:rzsx}
\end{equation}
where $\bf \Psi=\frac{1}{N} (\bar{A} \otimes A)\tilde{B} F^{H}$ is of size $M^2 \times N$. We recall that $\bf {r}_z$ is a vector of size $M^2$ and $\bf s_x$ is a vector of size $N$.

Since $\bf F$ is invertible, $\text{rank}(\mathbf{\Phi}) = \text{rank}((\mathbf{\bar{A} \otimes A)\tilde{B}})$. Note that we can express $\mathbf{C}=(\mathbf{\bar{A} \otimes A)\tilde{B}}$ as
\begin{eqnarray} \nonumber
\mathbf{C} = \left[ \begin{array}{c}
\mathbf{\bar{a}_1} \otimes \mathbf{a_1} + \mathbf{\bar{a}_2} \otimes \mathbf{a_2} + \dots + \mathbf{\bar{a}_N} \otimes \mathbf{a_N} \\
\mathbf{\bar{a}_1} \otimes \mathbf{a_2} + \mathbf{\bar{a}_2} \otimes \mathbf{a_1} + \mathbf{\bar{a}_2} \otimes \mathbf{a_3} + \dots + \mathbf{\bar{a}_N} \otimes \mathbf{a_{N-1}} \\
\vdots \\
\mathbf{\bar{a}_1} \otimes \mathbf{a_N} + \mathbf{\bar{a}_N} \otimes \mathbf{a_1}
\end{array} \right]^T,
\end{eqnarray}
where $\mathbf{a}_j$ denotes the $j$th column of $\bf A$. Analyzing conditions for $\bf C$ to be full rank does not appear to be straightforward. We therefore propose instead to investigate the following frequency domain approach.

\subsubsection{Frequency domain}
From (\ref{xmodel}),
\begin{equation}
\mathbf{c}= \mathbf{F} \mathbf{x}.
\end{equation}
Here, $\mathbf{x}$ is given by (\ref{xsamp}), the entries of $\mathbf{c}$ are the Fourier coefficients of $\bf x$ (see (\ref{xmodel})) and $\bf F$ is the $N \times N$ DFT matrix.
Since $\bf F$ is orthogonal,
\begin{equation}
\label{eq:xc}
\mathbf{x}= \frac{1}{N} \mathbf{F}^{H} \mathbf{c}.
\end{equation}
Define the autocorrelation matrix $\mathbf {R_c}=\mathbb{E} \left[ \mathbf{c} \mathbf{c}^H \right]$. From (\ref{eq:spec3}), $\mathbf{R_c}$ is a diagonal matrix and it holds that $\mathbf{R_c}(i,i)=s_{\mathbf{x}}[i-(N+1)/2]$.
Clearly, our goal can be stated as recovery of $\mathbf{R_c}$, since once $\mathbf{R_c}$ is known, $\mathbf{s}_\mathbf{x}$ follows.

We now relate $\mathbf{R_c}$ to the correlation of the sub-Nyquist samples. A variety of different sub-Nyquist schemes can be used to sample $x(t)$ \cite{Mishali_multicoset, Mishali_theory,Laska2007}, even when its Fourier series is not sparse as we will show in Section \ref{sec:rate1}. Let $\bf{z} \in \mathbb{R}^M$ denote the vector of sub-Nyquist samples of $x(t), 0 \leq t < T$, sampled at rate $f_s$ with $f_s<N/T$. For simplicity, we assume that $M=f_sT<N$ is an integer.
We express the sub-Nyquist samples $\bf{z}$ in terms of the Nyquist samples $\bf{x}$ as
\begin{equation}
\bf{z} = \bf{A} \bf{x},
\label{eq:dmodel1}
\end{equation}
where $\mathbf{A}$ is a $M \times N$ matrix. Combining (\ref{eq:xc}) and (\ref{eq:dmodel1}), we obtain
\begin{equation}
\label{eq:zc}
 \mathbf{z} = \frac{1}{N} \mathbf{A} \mathbf{F}^{H} \mathbf{c} \triangleq \mathbf{G} \mathbf{c},
\end{equation}
where $\mathbf{G}= \frac{1}{N}\mathbf{A} \mathbf{F}^{H}$. We assume that $\bf G$ is full spark, namely $\text{spark}(\mathbf{G})=M+1$.

Let  ${\mathbf{R}_\mathbf{z}} = \mathbb{E} \left[ \mathbf{z}\mathbf{z}^{H} \right]$ be the covariance matrix of the sub-Nyquist samples.
We now relate $\mathbf{R}_\mathbf{z}$ to $\mathbf{R}_\mathbf{c}$. From (\ref{eq:zc}), we have
\begin{equation}
\bf R_z = GR_cG^H.
\label{eq:autoco1}
\end{equation}
Vectorizing both sides of (\ref{eq:autoco1}),
\begin{equation}
\mathbf{r_z} = \mathbf{(\bar{G} \otimes G)}\text{vec}(\mathbf{R}_\mathbf{c}) =  \mathbf{(\bar{G} \otimes G)} \mathbf{B} \mathbf{r_c} \triangleq \mathbf{\Phi \mathbf{r_c}}.
\label{eq:rzrc}
\end{equation}
Here, $\bf B$ is as defined in Section \ref{analog_rec}, $\bf \Phi=(\bar{G} \otimes G) \mathbf{B}$ is of size $M^2 \times N$ and $\bf r_c$ is a vector of size $N$ that contains the potentially non-zero elements, namely the diagonal elements, of $\mathbf{R_c}$, that is $\mathbf{r_c}(i)=\mathbf{R_c}(i,i), 1 \le i \le N$.

In the second and third scenarii, $\bf r_c$ has only $K_f \ll N$ non zero elements, which correspond to the $K_f$ non zero Fourier cofficients in $S_x$. In Section \ref{sec:rate}, we discuss the conditions for (\ref{eq:rzrc}) and (\ref{eq:rzsx}) to have a unique solution, and we derive the minimal sampling rate accordingly.

Again, we assume full knowledge of $\mathbf{r_z}$ and will show how it can be approximated in Section \ref{sec:simulations}.

We observe that we obtain a similar relation (\ref{eq:rzrx}) and (\ref{eq:rzrc}) in both models. Therefore, the next section refers to both together. We used the same notation for different parameters in the two models so that they both lead to the same relation. In order to avoid confusion, we summarize the notation in \mbox{Table \ref{table:parameters}}. 
\begin{table}
\begin{tabular}{l|l|l} 
Parameter & Analog Model & Discrete Model \\
\hline
$M$ & $\#$ measurements & $\#$ sub-Nyquist samples \\
$N$ & $\#$ frequency bins & $\#$ Nyquist samples \\
$K_f$ & $\#$ potentially non & $\#$ potentially non zero \\
& zero frequency bands &  fourier coefficients \\
$S_x$ & continuous support & discrete Fourier series \\
& of $x(t)$ & support
\end{tabular}
\caption{Parameters notation in both models}
 \label{table:parameters}
\end{table}
We also note that, in the analog model, we define an infinite number of equations (\ref{eq:rzrx}), more precisely one per frequency ($f \in \mathcal{F}_s$), whereas in the digital model we obtain a single equation (\ref{eq:rzrc}).

\section{Minimal Sampling Rate and Reconstruction}
\label{sec:rate}

\subsection{No sparsity Constraints}
\label{sec:rate1}
\subsubsection{Minimal Rate for Perfect Reconstruction}

The systems defined in (\ref{eq:rzrx}) and (\ref{eq:rzrc}) are overdetermined for $M^2 \geq N$, if $\bf \Phi$ is full column rank. The following proposition provides conditions for the systems defined in (\ref{eq:rzrx}) and (\ref{eq:rzrc}) to have a unique solution.

\begin{proposition} \label{prop:first}
Let $\bf T$ be a full spark $M \times N$ matrix ($M \le N$) and $\bf B$ be a $N^2 \times N$ selection matrix that has a 1 in the $j$th column and $[(j-1)N+j]$th row, $1 \le j \le N$ and zeros elsewhere. The matrix $\bf C = \mathbf{(\bar{T} \otimes T)} \mathbf{B}=\bar{T} \odot T$ is full column rank  if $M^2 \ge N$ and $2M > N$.
\end{proposition}
\begin{proof}
First, we require $M^2 \ge N$ in order for $\bf C$ to have a smaller or equal number of columns than of rows.

Let $\bf x$ be a vector of length $N$ in the null space of $\bf C$, namely $\bf Cx=0$. We show that if $2M > N$, then $\bf x =0$.
Assume by contradiction that $\bf x \neq 0$. We denote by $\Omega_z$ the set of indices $1 \le j \le N$ such that $x_j \neq 0$ and $N_z=| \Omega_z |$. It holds that $1 \le N_z \le N$.

Note that we can express $\bf C$ as
\begin{eqnarray} \nonumber
\mathbf{C} &=& \left[ \begin{array}{cccc}
\mathbf{\bar{t}_1} \otimes \mathbf{t_1} & \mathbf{\bar{t}_2} \otimes \mathbf{t_2} & \dots & \mathbf{\bar{t}_N} \otimes \mathbf{t_N}
\end{array} \right],
\end{eqnarray}
where $\mathbf{t}_j$ denotes the $j$th column of $\bf T$. Let
\begin{equation} \nonumber
\mathbf{h}_i = \left[ \begin{array}{cccc} t_{i1} x_1 & t_{i2} x_2 & \dots & t_{iN} x_N
\end{array} \right]^T, \quad 1 \le i \le M.
\end{equation}
Then $\bf Cx=0$ if and only if
\begin{equation}
\label{eq:proof1}
\mathbf{\bar{T}} \mathbf{h}_i = 0, \quad 1 \le i \le M.
\end{equation}
That is, the $M$ vectors $\mathbf{h}_i$ are in the null space of $\mathbf{\bar{T}}$.

If $N_z \le M$, then since $\bf \bar{T}$ is full spark, (\ref{eq:proof1}) holds if and only if $t_{ij}x_j=0, \forall 1 \le i \le M \text{ and } \forall j \in \Omega_z$. Again, since $\mathbf{T}$ is full spark, none of its columns is the zero vector and therefore that $x_j=0, \forall j \in \Omega_z$ and we obtained a contradiction. 

If $N_z > M$, then we show that the vectors $\mathbf{h}_i, 1 \le i \le M$ are linearly independent. Since $\mathbf{T}$ is full spark, every set of $M$ columns are linearly independent. Let us consider $M$ columns $\mathbf{t}_j$ of $\bf T$ such that $ j \in \Omega_z$. It follows that
\begin{equation} \nonumber
\sum_j \gamma_j x_j \mathbf{t}_j=0
\end{equation}
if and only if $\gamma_j x_j = 0$. From the definition of $\Omega_j$, this holds if and only if $\gamma_j = 0$, that is the $M$ vectors $x_j\mathbf{a}_j$ are linearly independent. Thus, the $M$ vectors $\mathbf{h}_i$ are linearly independent as well.
We denote by $\text{nullity}(\mathbf{T})$ the dimension of the null space of $\bf T$. From the rank-nullity theorem, $\text{nullity}(\mathbf{\bar{T}})=N - \text{rank}(\mathbf{\bar{T}})=N-M$. Since the dimension of the space spanned by $\mathbf{h}_i$ is $M$, if $M>N-M$, then $\bf x=0$.
\end{proof}

The following theorem follows directly from Proposition \ref{prop:first}.
\begin{theorem}
The systems (\ref{eq:rzrx}) (analog model) and (\ref{eq:rzrc}) (digital model) have a unique solution if
\begin{enumerate}
\item $\bf A$ in the analog model and $\bf AF^H$ in the digital model are full spark.
\item $M^2 \ge N$ and $2M > N$.
\end{enumerate}
\end{theorem}

This can happen even for $M<N$ which is our basic assumption. If $M \ge 2$, we have $M^2 \ge 2M$. Thus, in this case, the values of $M$ for which we obtain a unique solution are $N/2 < M < N$. The minimal sampling rate is then
\begin{equation}
f_{(1)}=Mf_s>\frac{N}{2}B=\frac{f_{\text{Nyq}}}{2}.
\end{equation}

This means that even without any sparsity constraints on the signal, we can retrieve its power spectrum by exploiting its stationarity property, whereas the measurement vector $\bf z$ exhibits no stationary constraints in general. This was already observed in \cite{TianLeus} for the digital model, but no proof was provided. In \cite{wang}, the authors show that $2M>N$ is a sufficient condition on $M$ so that $\bf \Phi$ is full column rank in the analog model. Then, a universal sampling pattern can guarantee the full column rank of $\bf \Phi$.
In \cite{Davies, Davies2, wang}, the authors claim that the system is overdetermined if $M(M-1)+1 \ge N$ and if the multicoset sampling pattern is such that it yields a full column rank matrix $\bf \Phi$. In \cite{wang}, some simple sub-optimal multicoset sampling patterns are given, that achieve compression rate below $1/2$. Some examples of optimal patterns, namely that guarantee a unique solution under $M(M-1)+1 = N$, are given in \cite {Davies, Davies2} but it is not clear what condition is required from the pattern, or alternatively from the sampling matrix $\bf A$, in order for $\bf \Phi$ to have full column rank. Here, the condition for having a solution is given with respect to the sampling matrix $\bf A$, which directly depends on the sampling pattern, rather than the matrix $\bf \Phi$.

\subsubsection{Power Spectrum Reconstruction}
If the conditions of Proposition \ref{prop:first} are satisfied, namely if the sampling rate $f \ge f_{(1)}$, then the systems defined in (\ref{eq:rzrx}) and (\ref{eq:rzrc}) are overdetermined, respectively. The power spectrums $\mathbf{r_x}(f)$ and $\bf r_c$ are given by
\begin{equation} \label{eq:rec11}
\mathbf{\hat{r}_x}(f) = \mathbf{\Phi}^{\dagger} \mathbf{r_z}(f),
\end{equation}
in the analog model, and
\begin{equation} \label{eq:rec12}
\bf \hat{r}_c = \Phi^{\dagger} r_z,
\end{equation}
in the digital one. Here $\dagger$ denotes the Moore-Penrose pseudo-inverse.

%

\subsection{Sparsity Constraints - Non-Blind Detection}
\label{rate2}

\subsubsection{Minimal Rate for Perfect Reconstruction}
We now consider the second scheme, where we have \emph{a priori} knowledge on the frequency support of $x(t)$ and we assume that it is sparse. 
Instead of reconstructing the entire power spectrum, we exploit the knowledge of the signal's frequencies in order to recover the potentially occupied bands (analog model) or the potential Fourier series coefficients of the autocorrelation function (discrete model). This will allow us to further reduce the sampling rate.

In this scenario, $\mathbf{r_x}(f)$ (first model) and $\bf r_c$ (second model) contain only $K_f \ll N$ potentially non zero elements as discussed in Section \ref{SecOpt}. In the first model, the reduced problem can be expressed as
\begin{equation}
\mathbf{r_z}(f) = \mathbf{\Phi}_S\mathbf{r}_{\mathbf{x}}^S(f).
\label{eq:rzrx2}
\end{equation}
Here, $\mathbf{r}_{\mathbf{x}}^S(f)$ is the vector $\mathbf{r_x}(f)$ reduced to its $K_f$ potentially non zero elements and $\mathbf{\Phi}_S$ contains the corresponding $K_f$ columns of $\bf \Phi$.
Here, the support $S$ of $\mathbf{r_x}(f)$ depends on the specific frequency $f$ since the support of the power spectrum of each transmission $s_i(t)$ can be split into two different bins of $\mathbf{r_x}(f)$. Obviously, $S$ can be calculated for each $f$ from the known $S_x$.

In the second model, the reduced problem becomes
\begin{equation}
\mathbf{r_z} = \mathbf{\Phi}_S\mathbf{r}_{\mathbf{c}}^S.
\label{eq:rzrc2}
\end{equation}
Here, $\mathbf{r}_{\mathbf{c}}$ is the reduction of $\bf r_c$ to its $K_f$ potentially non zero elements, $\mathbf{\Psi}_S$ contains the corresponding $K_f$ columns of $\bf \Psi$.
In the digital case, it holds that the support $S=S_x$.

The following proposition provides conditions for the systems defined in (\ref{eq:rzrx2}) and (\ref{eq:rzrc2}) to have a unique solution.
\begin{proposition} \label{prop:second}
Let $\bf T$ be a full spark $M \times N$ matrix ($M \le N$) and $\bf B$ be defined as in Proposition \ref{prop:first}. Let $\bf C = \mathbf{(\bar{T} \otimes T)}B$ and $\mathbf{H}$ be the $N \times K_f$ that selects any $K_f < N$ columns of $\bf C$. The matrix $\bf D=CH$ is full column rank if $M^2 \ge K_f$ and $2M> K_f$.
\end{proposition}
\begin{proof}
First, we require $M^2 \ge K_f$ in order for $\bf D$ to have a smaller or equal number of columns than of rows.
Let $\mathbf{T}_S$ be the $M \times K_f$ matrix composed of the $K_f$ columns of $\bf T$ corresponding to the $K_f$ selected columns of $\bf C$:
\begin{equation} \nonumber
\mathbf{T}_S = \left[ \begin{array}{cccc} 
\mathbf{t_{[1]}} &  \mathbf{t_{[2]}} & \dots & \mathbf{t_{[K_f]}}
\end{array} \right].
\end{equation}
Here $\mathbf{t}_{[i]}, 1 \le i \le K_f$ denotes the column of $\bf t$ corresponding to the $i$th selected column of $\bf C$.
We have 
\begin{equation} \nonumber
\label{eq:matrixD}
\mathbf{D} = \left[ \begin{array}{cccc}
\mathbf{\bar{t}_{[1]}} \otimes \mathbf{t_{[1]}} & \mathbf{\bar{t}_{[2]}} \otimes \mathbf{t_{[2]}} & \dots & \mathbf{\bar{t}_{[K_f]}} \otimes \mathbf{t_{[K_f]}}
\end{array} \right].
\end{equation}

If $K_f \ge M$, then, since $\bf \text{spark}(T)$$ = M+1$, $\mathbf{T}_S$ is full spark as well. Applying Proposition \ref{prop:first} with $\mathbf{T}_S$, we have that $\bf D$ is full column rank if $2M> K_f$.

If $K_f<M$, then from $\bf \text{spark}(T)$$ = M+1>K_f$, $\mathbf{T}_S$ is full column rank. Since $\text{rank} (\bar{\mathbf{T}}_S \otimes \mathbf{T}_S) = \text{rank} (\mathbf{T}_S)^2$$ = K_f^2$, the matrix $\bar{\mathbf{T}}_S \otimes \mathbf{T}_S$ is also full column rank. It can be seen that the matrix $\bf D$ is obtained by selecting $K_f$ columns from $\bar{\mathbf{T}}_S \otimes \mathbf{T}_S$. It follows that $\bf D$ is full column rank as well.
\end{proof}

The following theorem follows directly from Proposition \ref{prop:second}.
\begin{theorem}
The systems (\ref{eq:rzrx2}) (analog model) and (\ref{eq:rzrc2}) (digital model) have a unique solution if
\begin{enumerate}
\item $\bf A$ in the analog model and $\bf AF^H$ in the digital model are full spark.
\item $M^2 \ge K_f$ and $2M > K_f$.
\end{enumerate}
\end{theorem}

In this case, the minimal sampling rate is
\begin{equation}
f_{(2)}=Mf_s>\frac{K_f}{2}B=N_{\text{sig}}B.
\end{equation}
Landau \cite{LandauCS} developed a minimal rate requirement for perfect signal reconstruction in the non-blind setting, which corresponds to the actual band occupancy, namely $2N_{\text{sig}}B$. Here, we find that the minimal sampling rate for perfect spectrum recovery in this setting is half the Landau rate.

\subsubsection{Power Spectrum Reconstruction}
If the conditions of Proposition \ref{prop:second} are satisfied, namely the sampling rate $f \ge f_{(2)}$, then we can reconstruct the signal's power spectrum by first reducing the systems as shown in (\ref{eq:rzrx2}) and (\ref{eq:rzrc2}). The reconstructed power spectrum $\mathbf{r_x}(f)$ and $\bf r_c$ are given by
\begin{eqnarray}
\mathbf{\hat{r}_x}^S(f) &=& \mathbf{\Phi}_S^{\dagger} \mathbf{r_z}(f) \\
\mathbf{\hat{r}}_{\mathbf{x}_i}(f) &=& 0 \quad \forall i \notin S, \nonumber
\end{eqnarray}
in the first model, and
\begin{eqnarray}
\mathbf{\hat{r}_c} &=&\mathbf{\Phi}_S^{\dagger} \mathbf{r_z} \\
\mathbf{\hat{r}}_{\mathbf{c}_i}&=& 0 \quad \forall i \notin S, \nonumber
\end{eqnarray}
in the second one.

\subsection{Sparsity Constraints - Blind Detection}
\label{sec:scen3}
\subsubsection{Minimal Rate for Perfect Reconstruction}

We now consider the third scheme, namely $x(t)$ is sparse, without any \emph{a priori} knowledge on the support. In the previous section, we showed that $\mathbf{\Phi}_S$ is full column rank, for any choice of $K_f<2M$ columns of $ \bf \Phi$. Thus, for $M \ge 2$, we have $\text{spark}(\mathbf{\Phi}) \ge 2M$. Therefore, in the blind setting, if $\mathbf{r_x}(f)$ or $\bf r_c$ is $K_f$-sparse, with $K_f < M$, it is the unique sparsest solution to (\ref{eq:rzrx}) or (\ref{eq:rzrc}), respectively \cite{CSBook}.
In this case, the minimal sampling rate is
\begin{equation}
\label{eq:min3}
f_{(3)}= Mf_s>K_fB=2N_{\text{sig}}B,
\end{equation}
which is twice the rate obtained in the previous scenario. As in signal recovery, the minimal rate for blind reconstruction is twice the minimal rate for non-blind reconstruction \cite{Mishali_multicoset}. 

The authors in \cite{Davies2} consider the sparse case as well for a model similar to our analog model. Again, the conditions for the system to be overdetermined are given with respect to $\bf \Phi$, as in the non sparse case. Moreover, the authors reconstruct the average spectrum of the signal over each bin, rather than the spectrum itself at each frequency. Here, the two approaches become fundamentally different since in this scenario, we deal with a system of equations of infinite measure whereas in \cite{Davies2}, the authors obtain a standard compressed sensing problem aiming at recovering a finite vector.

\subsubsection{Power Spectrum Reconstruction}
In this scenario, there exists an inherent difference between the two models. In the digital model, we have to solve a single equation (\ref{eq:rzrc}) whereas in the analog model, (\ref{eq:rzrx}) consists of an infinite number of linear systems because $f$ is a continuous variable.

Therefore, in the digital case, we can use classical compressed sensing (CS) techniques \cite{CSBook} in order to recover the sparse vector $\bf r_c$ from the measurement vector $\bf r_z$, namely
\begin{equation}
\mathbf{\hat{r}_c} = \arg\min_{\mathbf{r_c}} || \mathbf{r_c} ||_0 \quad \text{s.t. } \mathbf{r_z=\Phi r_c}.
\end{equation}

In the analog model, the reconstruction can be divided into two stages: support recovery and spectrum recovery. We use the support recovery paradigm from \cite{Mishali_multicoset} that produces a finite system of equations, called multiple measurement vectors (MMV) from an infinite number of linear systems. This reduction is performed by what is referred to as the continuous to finite (CTF) block. From (\ref{eq:rzrx}), we have
\begin{equation}
\bf Q = \Phi Z \Phi^H
\end{equation}
where
\begin{equation}
\mathbf{Q}= \int_{f \in \mathcal{F}_s} \mathbf{r_z}(f) \mathbf{r_z}^H(f) \mathrm{d}f
\end{equation}
is a $M \times M$ matrix and
\begin{equation}
\mathbf{Z}= \int_{f \in \mathcal{F}_s} \mathbf{r_x}(f) \mathbf{r_x}^H(f) \mathrm{d}f
\end{equation}
is a $N \times N$ matrix. We then construct a frame $\bf V$ such that $\bf Q=VV^H$. Clearly, there are many possible ways to select $\bf V$. We choose to construct it by performing an eigendecomposition of $\bf Q$ and select $\bf V$ as the matrix of eigenvectors corresponding to the non zero eigenvalues. We can then define the following linear system
\begin{equation} \label{eq:CTF}
\bf V= \Phi U
\end{equation}
From \cite{Mishali_multicoset} (Propositions 2-3), the support of the unique sparsest solution of (\ref{eq:CTF}) is the same as the support of our original set of equations (\ref{eq:rzrx}).

As discussed in Section \ref{SecOpt}, $\mathbf{r_x}(f)$ is $K_f$-sparse for each specific $f \in \mathcal{F}_s$. However, after combining the frequencies, the matrix $\bf U$ is $2K_f$-sparse (at most), since the spectrum of each transmission can be split into two bins of $\mathbf{r_x}(f)$. Therefore, the above algorithm, referred to as SBR4 in \cite{Mishali_multicoset} (for signal reconstruction as opposed to spectrum reconstruction), requires a minimal sampling rate of $2f_{(3)}$. In order to achieve the minimal rate $f_{(3)}$, the SBR2 algorithm regains the factor of two in the sampling rate at the expense of increased complexity \cite{Mishali_multicoset}. In a nutshell, SBR2 is a recursive algorithm that alternates between the CTF described above and a bi-section process. The bi-section splits the original frequency interval into two equal width intervals on which the CTF is applied, until the level of sparsity of $\mathbf{U}$ is less or equal to $K_f$. As opposed to SBR4 which can be performed both in the time and in the frequency domains, SBR2 can obviously be performed only in the frequency domain. We refer the reader to \cite{Mishali_multicoset} for more details.

Once the support $S$ is known, perfect reconstruction of the spectrum can be obtained as follows
\begin{eqnarray} \label{eq:recs}
\mathbf{\hat{r}_x}^S(f) &=& \mathbf{\Phi}_S^{\dagger} \mathbf{r_z}(f) \\
\mathbf{\hat{r}}_{\mathbf{x}_i}(f) &=& 0 \quad \forall i \notin S. \nonumber
\end{eqnarray}

\section{Simulation Results}
\label{sec:simulations}

We now demonstrate power spectrum reconstruction from sub-Nyquist samples for the first and third scenarii, respectively. We also investigate the impact of several simulation parameters on the receiver operating characteristic (ROC) of our detector: signal-to-noise ratio (SNR), sensing time, number of averages (for estimating the autocorrelation matrix $\mathbf{R_z}$ as explained below) and sampling rate. Last, we compare the performance of our detector to one based on spectrum reconstruction from sub-Nyquist samples and a second one based on power spectrum reconstruction from Nyquist samples. Throughout the simulations we consider the analog model and use the MWC analog front-end for the sampling stage.

\subsection{Reconstruction in time and frequency domains and detection}
We first explain how we estimate the elements of $\bf r_z$. The overall sensing time is divided into $P$ frames of length $K$ samples. In Section \ref{subsec:param}, we examine different choices of $P$ and $K$ for a fixed sensing time.
In the digital model, the estimate of $\bf r_z$ is simply obtained by averaging the autocorrelation between the samples $\bf z$ over $P$ frames as follows
\begin{equation}
\mathbf{\hat{r}_z}=\frac{1}{P} \sum_{p=1}^P \mathbf{z}^p (\mathbf{z}^p)^H
\end{equation}
where $\mathbf{z}^p$ is the vector of sub-Nyquist samples of the $p$th frame.

In the analog model, in order to estimate the autocorrelation matrix $\mathbf{R_z}(f)$ in the frequency domain, we first compute the estimates of $\mathbf{z}_i(f), 1 \le i \le M$, $\hat{\mathbf{z}}_i(f)$, using FFT on the samples $z_i[n]$ over a finite time window. We then estimate the elements of $\mathbf{R_z}(f)$ as
\begin{equation}
\mathbf{\hat{R}_z}(i,j,f)=\frac{1}{P} \sum_{p=1}^{P} \hat{\mathbf{z}}^p(i,f) \hat{\mathbf{z}}^p(j,f), \quad f \in \mathcal{F}_s,
\end{equation}
where $P$ is the number of frames for the averaging of the spectrum and $\hat{\mathbf{z}}^p(i,f)$ is the value of the FFT of the samples $\mathbf{z}_i[n]$ from the $p$th frame, at frequency $f$.
In order to estimate the autocorrelation matrix $\mathbf{R_z}[n]$ in the time domain, we convolve the samples $z_i[n]$ over a finite time window as
\begin{equation}
\mathbf{\hat{R}_z}[i,j,n]=\frac{1}{P} \sum_{p=1}^{P} z_i^p[n] * z_j^p[n], \quad n \in \mathcal [0, T/T_{\text{Nyq}}].
\end{equation}
We then use (\ref{eq:rec11}) or (\ref{eq:recs}) in order to reconstruct $\mathbf{\hat{r}_x}(f)$, or their time-domain equivalents to reconstruct $\mathbf{\hat{r}_x}[n]$.

We note that the number of samples dictates the number of DFT cofficients in the frequency domain and therefore the resolution of the reconstructed spectrum in the frequency domain.

For the analog model, we use the following test statistic
\begin{equation}
\mathcal{T}_i = \sum |\hat{\mathbf{r}}_{\mathbf{x}_i}|, \qquad 1 \le i \le N,
\end{equation}
where the sum is performed over frequency or over time, depending on which domain we chose to reconstruct $\bf \hat{r}_x$. Obviously, other detection statistics can be used on the reconstructed power spectrum.

\subsection{Spectrum reconstruction}

We first consider spectrum reconstruction of a non sparse signal. Let $x(t)$ be white Gaussian noise with variance $100$, and Nyquist rate $f_{\text{Nyq}}=10GHz$ with two stop bands. We consider $N=65$ spectral bands and $M=33$ analog channels, each with sampling rate $f_s=154MHz$ and with $N_s=131$ samples each. The overall sampling rate is therefore equal to $50.77 \%$ of the Nyquist rate. Figure \ref{fig:sim1} shows the original and the reconstructed spectrum at half the Nyquist rate (both with averaging over $P=1000$).

\begin{figure}[!h]
\begin{center}
\includegraphics[width=0.5\textwidth]{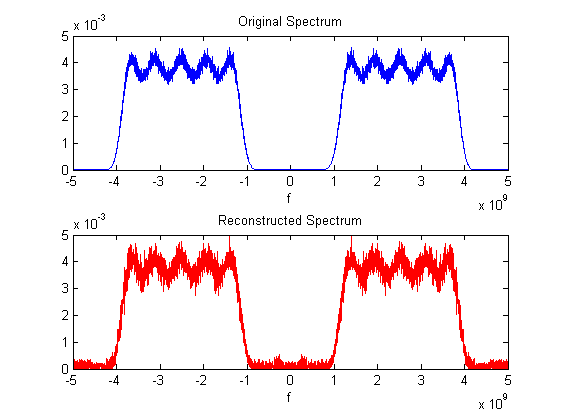}
\caption{Original and reconstructed spectrum of a non sparse signal at half the Nyquist rate.}
\label{fig:sim1}
\end{center}
\end{figure}

We now consider the blind reconstruction of the power spectrum of a sparse signal. Let the number of potentially active transmissions $N_{\text{sig}}=3$. Each transmission is generated from filtering white Gaussian noise with a low pass filter whose two-sided bandwidth is $B=80MHz$, and modulating it with a carrier frequency drawn uniformly at random between $-f_{\text{Nyq}}/2=-5GHz$ and $f_{\text{Nyq}}/2=5GHz$. We consider $N=65$ spectral bands and $M=7$ analog channels, each with sampling rate $f_s=154MHz$ and with $K=171$ samples per channel and per frame. The overall sampling rate is equal to $10.8 \%$ of the Nyquist rate, and $1.9$ times the Landau rate. We consider additive white Gaussian noise. Figures \ref{fig:sim21}-\ref{fig:sim26} shows the original and the reconstructed power spectrum for different values of the number of frames $P$ and of the SNR.

\begin{figure}[!h]
\begin{center}
\includegraphics[width=0.5\textwidth]{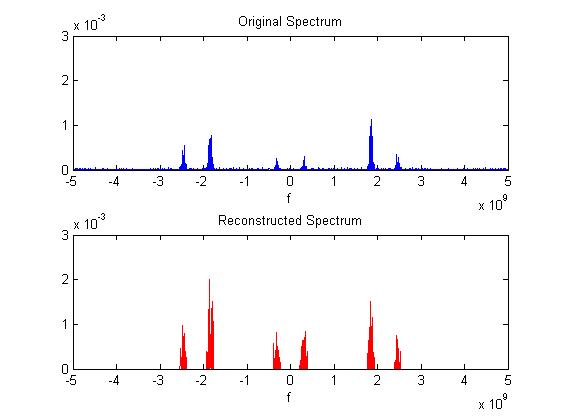}
\caption{Original and reconstructed spectrum: $P=1$ and SNR$=0$dB.}
\label{fig:sim21}
\end{center}
\end{figure}

\begin{figure}[!h]
\begin{center}
\includegraphics[width=0.5\textwidth]{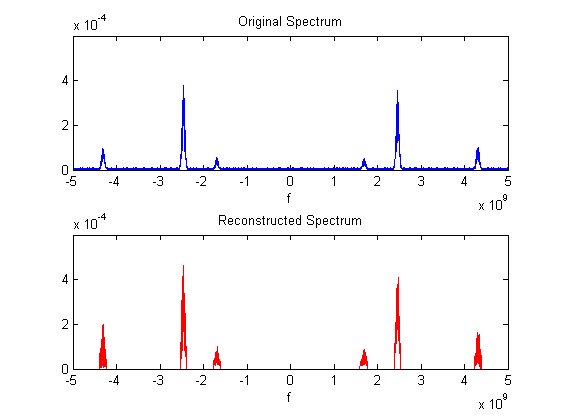}
\caption{Original and reconstructed spectrum: $P=25$ and SNR$=0$dB.}
\label{fig:sim22}
\end{center}
\end{figure}

\begin{figure}[!h]
\begin{center}
\includegraphics[width=0.5\textwidth]{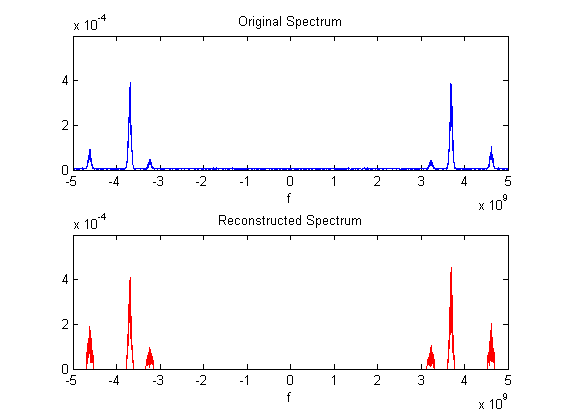}
\caption{Original and reconstructed spectrum: $P=50$ and SNR$=0$dB.}
\label{fig:sim23}
\end{center}
\end{figure}

\begin{figure}[!h]
\begin{center}
\includegraphics[width=0.5\textwidth]{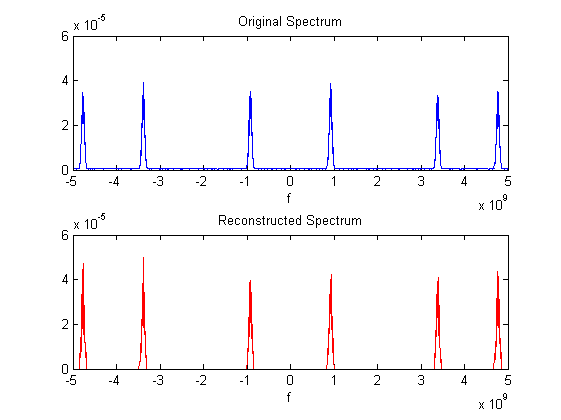}
\caption{Original and reconstructed spectrum: $P=100$ and SNR$=2$dB.}
\label{fig:sim24}
\end{center}
\end{figure}

\begin{figure}[!h]
\begin{center}
\includegraphics[width=0.5\textwidth]{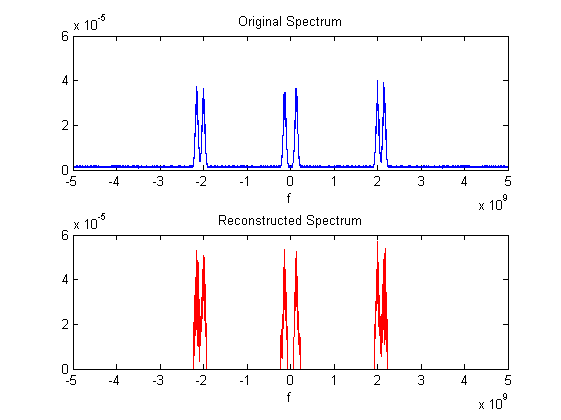}
\caption{Original and reconstructed spectrum: $P=100$ and SNR$=0$dB.}
\label{fig:sim25}
\end{center}
\end{figure}

\begin{figure}[!h]
\begin{center}
\includegraphics[width=0.5\textwidth]{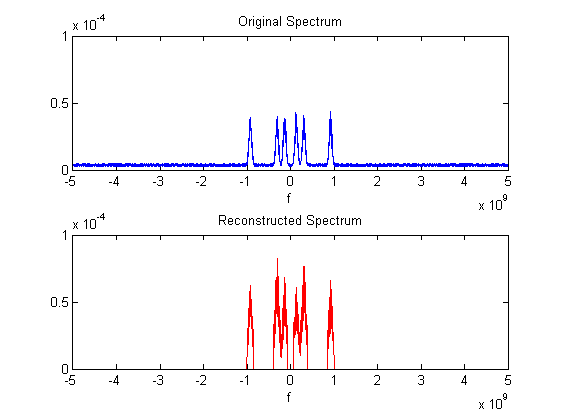}
\caption{Original and reconstructed spectrum: $P=100$ and SNR$=-2$dB.}
\label{fig:sim26}
\end{center}
\end{figure}


%

\subsection{Practical parameters}
\label{subsec:param}
In this section, we consider the influence of several practical parameters on the performance of our detector. The experiments are set up as follows. We consider two scenarios where the actual number of transmissions is 2 and 3, namely $\mathcal{H}_{0}$ and $\mathcal{H}_{1}$ respectively. The number of potentially active transmissions $N_{\text{sig}}$ is set to be 6. Each transmission is similar to those described in the previous experiment. We consider $N=115$ spectral bands and $M$ analog channels, each with sampling rate $f_s=87MHz$. The number of samples per channel and per frame is $K$ and the averaging is performed over $P$ frames. Each experiment is repeated over $500$ realisations.

In the first experiment, we illustrate the impact of SNR on the detection performance. We consider $M=8$  channels. The overall sampling rate is thus $695MHz$, which is a little below $7\%$ of the Nyquist rate and a little above $1.2$ times the Landau rate. Here, $K=171$ and $P=10$ frames. Figure \ref{fig:sim3} shows the ROC of the detector for different values of SNR. 
\begin{figure}[!h]
\begin{center}
\includegraphics[width=0.5\textwidth]{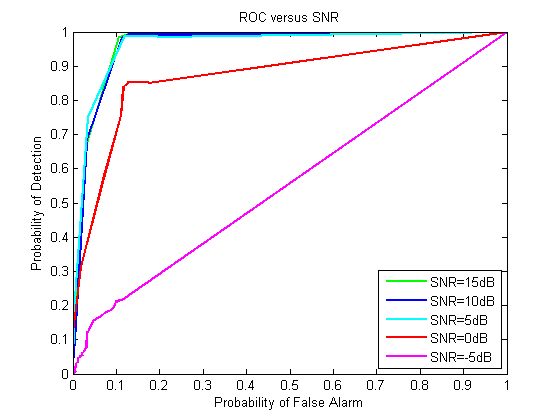}
\caption{Influence of the SNR on the ROC.}
\label{fig:sim3}
\end{center}
\end{figure}
We observe that up to a certain value of the SNR (between $5$dB and $0$dB in this setting), the detection performance does not decrease much. Below that, the performance decreases rapidly. Another observation that can be made concerns the particular form of the ROC curves. These can be split into two parts. The first part corresponds to a regular ROC curve, where the probability of detection increases faster than linearly with the probability of false alarm. After a certain point, the increase becomes linear. This corresponds to the realisations where the support recovery failed and the energy measured in the band of interest is zero both for $\mathcal{H}_{0}$ and $\mathcal{H}_{1}$. The more such realisations there are, the lower the point where the curve's nature changes. As one can expect, this point is lower for lower SNRs.

In the second experiment, we vary the sensing time per frame and keep the number of frames $P=10$ constant. We consider the same sampling parameters as in the previous experiment and set the SNR to be $2$dB. Figure \ref{fig:sim4} shows the ROC of the detector for different values of the number of samples per frame.
\begin{figure}[!h]
\begin{center}
\includegraphics[width=0.5\textwidth]{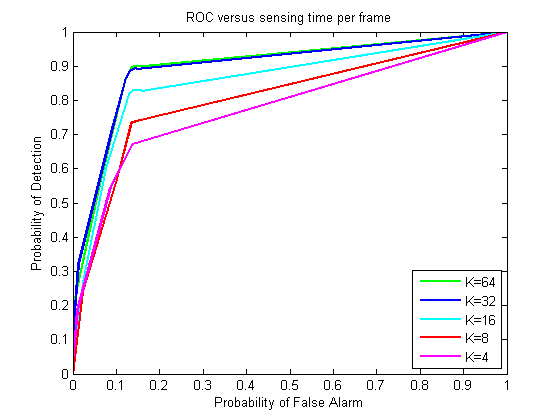}
\caption{Influence of the number of samples per frame on the ROC.}
\label{fig:sim4}
\end{center}
\end{figure}

In the third experiment, we vary the number of frames and keep the number of samples per frame $K=20$ constant. We consider the same sampling parameters as above and set the SNR to be $0$dB. Figure \ref{fig:sim5} shows the ROC of the detector for different values of the number of frames.
\begin{figure}[!h]
\begin{center}
\includegraphics[width=0.5\textwidth]{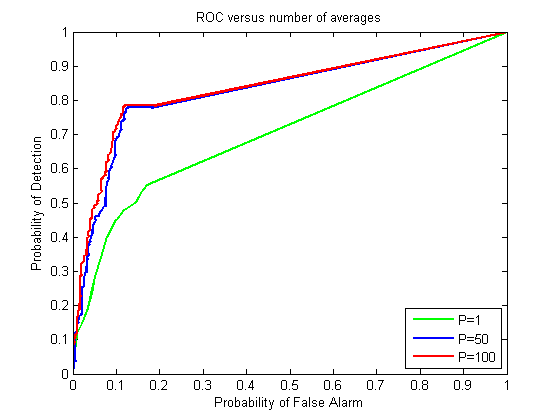}
\caption{Influence of the number of frames on the ROC.}
\label{fig:sim5}
\end{center}
\end{figure}
We observe that above a certain threshold, increasing the number of averages $P$ almost does not affect the detection performance.

An interesting question is, given a limited overall sensing time, or equivalently a limited number of samples, how should one set the number of frames $P$ and the number of samples per frame $K$. In the next experiment, we investigate different choices of $P$ and $K$ for a fixed total number of samples per channel $PK=100$. The rest of the parameters remain unchanged.  Figure \ref{fig:sim6} shows the ROC of the detector for those different settings.
\begin{figure}[!h]
\begin{center}
\includegraphics[width=0.5\textwidth]{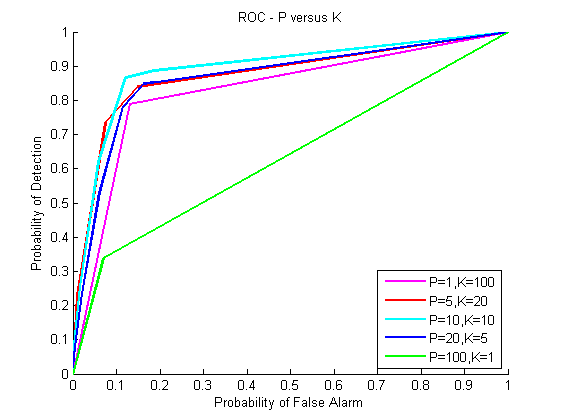}
\caption{Trade-off between the number of frames and the number of samples per frame.}
\label{fig:sim6}
\end{center}
\end{figure}
We can see that in this case, the best performance is attained for a balanced division of the number of samples, namely $P=10$ frames with $K=10$ samples each.

Last, we show the impact of the number of channels, namely the overall sampling rate, on the performance of our detector. The sampling parameters are set as above and the SNR is $0$dB. Figure \ref{fig:sim6} shows the ROC of the detector for different values of the number of channels.
\begin{figure}[!h]
\begin{center}
\includegraphics[width=0.5\textwidth]{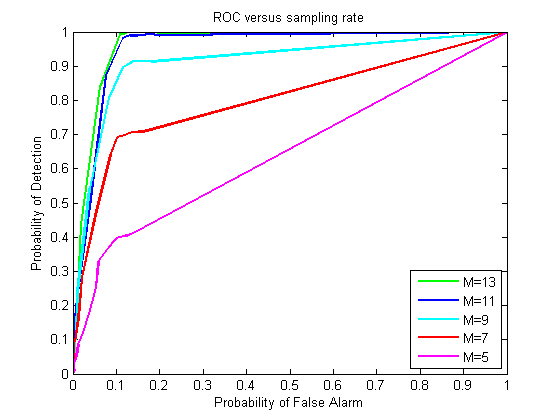}
\caption{Influence of the sampling rate on the ROC.}
\label{fig:sim7}
\end{center}
\end{figure}
The minimal number of channels in this case is $7$. Due to the noise presence, we need to sample above that threshold to obtain good detection performance. We observe that above $9$ channels, the performance increases very little with the number of channels whereas below the threshold of $7$ it decreases drastically.

\subsection{Performance Comparaisons}
We now compare our approach to sub-Nyquist spectrum sensing and nyquist power spectrum sensing.
\subsubsection{Power Spectrum versus Spectrum Reconstruction}

\label{comp_mishali}
First, we consider the approach of \cite{Mishali_theory} where the signal itself is reconstructed from sub-Nyquist samples. We compute the energy of the frequency band of interest and compare this spectrum based detection to our power spectrum based detection. We consider the exact same signal as in the previous section. The sampling parameters are as follows: $N=115$ spectral bands and $M=12$ analog channels, each with sampling rate $f_s=87MHz$. We recall that the minimal sampling rate for signal recovery is twice that needed for power spectrum recovery. The overall sampling rate is therefore $1.04GHz$, namely a little above $10\%$ of the Nyquist rate and almost $1.9$ times the Landau rate. The number of samples per channel and per frame is $K=10$ and the averaging is performed over $P=10$ frames. In the signal reconstruction approach, no averaging needs to be performed. Therefore, we use a total of $PK=100$ samples. Each experiment is repeated over $500$ realisations. Figure \ref{fig:sim31} shows the ROC of both detectors for different values of SNR. 
\begin{figure}[!h]
\begin{center}
\includegraphics[width=0.5\textwidth]{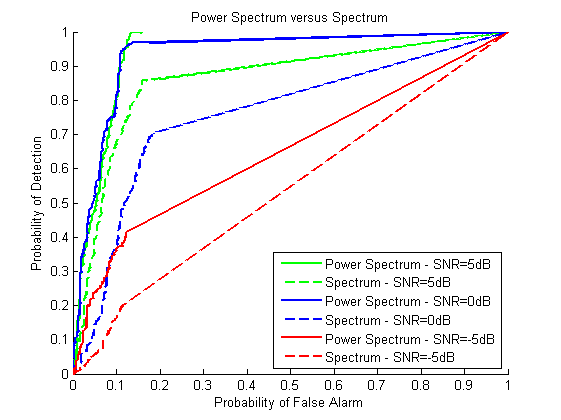}
\caption{Power spectrum versus spectrum reconstruction.}
\label{fig:sim31}
\end{center}
\end{figure}
We oberve that power spectrum sensing outperforms spectrum sensing.

\subsubsection{Nyquist versus Sub-Nyquist Sampling}
We now compare our approach to power spectrum sensing from Nyquist rate samples. We consider the exact same signal and sampling parameters as in Section \ref{comp_mishali} except for the number of channels which is set to $M=9$, leading to an overall sampling rate of $783GHz$, namely a little above $7.8\%$ of the Nyquist rate and $1.4$ times the Landau rate. Figure \ref{fig:sim32} shows the ROC of both detectors for different values of SNR. 
\begin{figure}[!h]
\begin{center}
\includegraphics[width=0.5\textwidth]{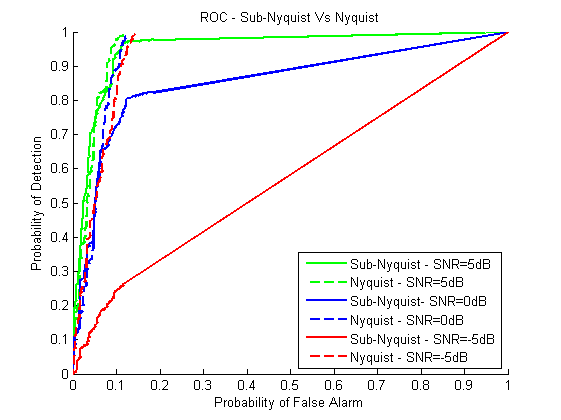}
\caption{Sub-Nyquist versus Nyquist sampling.}
\label{fig:sim32}
\end{center}
\end{figure}
It can be seen that our detector performs similarly as the Nyquist rate one up to a certain SNR threshold (around $5$dB in this setting). Below that threshold, the performance of our sub-Nyquist receiver decreases faster with SNR decrease whereas the Nyquist rate performance detection remains almost unchanged. The comes from the fact that the sensitivity of energy detection is amplified when performed on sub-Nyquist samples due to noise aliasing \cite{Castro}.

\section{Conclusion}
In this paper, we considered power spectrum reconstruction of stationary signals from sub-Nyquist samples. We investigated two signal models: the multiband model referred to as the analog model and the multi-tone model converted into a digital model. For the analog setting, two sampling schemes were adopted and for the digital one, two power spectrum reconstruction schemes were considered.  We showed that all variations of  both the analog and the digital models can be treated and analyzed in a uniform way in the frequency domain whereas a time domain analysis is a lot more complex. 

We derived the minimal sampling rate for perfect power spectrum reconstruction in noiseless settings for the cases of sparse and non sparse signals as well as blind and non blind detection. We also presented recovery techniques for each one of those scenarii. Simulations show power spectrum reconstruction at sub-Nyquist rates as well as the influence of practical parameters such as noise, sensing time and sampling rate on the ROC of the detector. We also showed that sub-Nyquist power spectrum sensing outperforms sub-Nyquist spectrum sensing and that our detector performance is comparable to that of a Nyquist rate power spectrum based detector up to a certain SNR threshold.

\section*{Acknowledgement}
The authors would like to thank Prof. Geert Leus for his valuable input and useful comments.

This work is supported in part by the Israel Science Foundation under Grant no. 170/10, in part by the Ollendorf Foundation, in part by the SRC, and in part by the Intel Collaborative Research Institute for Computational Intelligence (ICRI-CI).

\bibliographystyle{IEEEtran}
\bibliography{IEEEabrv,CompSens_ref}

\begin{thebibliography}{10}
\providecommand{\url}[1]{#1}
\csname url@samestyle\endcsname
\providecommand{\newblock}{\relax}
\providecommand{\bibinfo}[2]{#2}
\providecommand{\BIBentrySTDinterwordspacing}{\spaceskip=0pt\relax}
\providecommand{\BIBentryALTinterwordstretchfactor}{4}
\providecommand{\BIBentryALTinterwordspacing}{\spaceskip=\fontdimen2\font plus
\BIBentryALTinterwordstretchfactor\fontdimen3\font minus
  \fontdimen4\font\relax}
\providecommand{\BIBforeignlanguage}[2]{{%
\expandafter\ifx\csname l@#1\endcsname\relax
\typeout{** WARNING: IEEEtran.bst: No hyphenation pattern has been}%
\typeout{** loaded for the language `#1'. Using the pattern for}%
\typeout{** the default language instead.}%
\else
\language=\csname l@#1\endcsname
\fi
#2}}
\providecommand{\BIBdecl}{\relax}
\BIBdecl

\bibitem{Study1}
FCC, ``Spectrum policy task force report: Federal communications commission,
  tech. rep. 02-135. [online],''
  \emph{http://www.gov/edocs\_public/attachmatch/DOC\-228542A1.pdf}, Nov. 2002.

\bibitem{Study2}
M.~McHenry, ``{NSF} spectrum occupancy measurements project summary. shared
  spectrum co., tech. rep. [online],'' \emph{http://www.sharedspectrum.com},
  Aug. 2005.

\bibitem{study3}
R.~I.~C. Chiang, G.~B. Rowe, and K.~W. Sowerby, ``A quantitative analysis of
  spectral occupancy measurements for cognitive radio,'' \emph{Proc. of IEEE
  Vehicular Technology Conference}, pp. 3016--3020, Apr. 2007.

\bibitem{Mitola}
J.~Mitola, ``Software radios: Survey, critical evaluation and future
  directions,'' \emph{IEEE Aerosp. Electron. Syst. Mag}, vol.~8, pp. 25--36,
  Apr. 1993.

\bibitem{Haykin}
S.~Haykin, ``Cognitive radio: Brain-empowered wireless communications,''
  \emph{IEEE J. Select. Areas Commun.}, vol.~23, pp. 201--220, Feb. 2005.

\bibitem{MitolaMag}
J.~Mitola and C.~{Q. Maguire Jr.}, ``Cognitive radio: Making software radios
  more personal,'' \emph{IEEE Personal Commun.}, vol.~6, pp. 13--18, Aug. 1999.

\bibitem{Shannon}
C.~E. Shannon, ``Computers and automata,'' \emph{Proc. IRE}, vol.~41, pp.
  1234--1241, Oct. 1953.

\bibitem{cog}
A.~{Ghasemi} and E.~{S. Sousa}, ``Spectrum sensing in cognitive radio networks:
  requirements, challenges and design trade-offs,'' \emph{IEEE Communications
  Magazine}, vol.~46, pp. 32--39, Apr. 2008.

\bibitem{cognitive1}
E.~G. Larsson and M.~Skoglund, ``Cognitive radio in a frequency-planned
  environment: Some basic limits,'' \emph{IEEE Trans. Wireless Commun.},
  vol.~7, pp. 4800--4806, Dec. 2008.

\bibitem{cognitive2}
N.~H. A.~Sahai and R.~Tandra, ``Some fundamental limits on cognitive radio,''
  \emph{Proc. 42nd Annu. Allerton Conf. Communication, Control, and Computing},
  pp. 1662--1671, Oct. 2004.

\bibitem{WidebandMishali}
M.~{Mishali} and Y.~{C. Eldar}, ``Wideband spectrum sensing at sub-{N}yquist
  rates,'' \emph{{IEEE} Signal Process. Mag.}, vol.~28, no.~4, pp. 102--135,
  Jul. 2011.

\bibitem{Mishali_theory}
------, ``From theory to practice: Sub-{N}yquist sampling of sparse wideband
  analog signals,'' \emph{{IEEE} J. Sel. Topics Signal Process.}, vol.~4,
  no.~2, pp. 375--391, Apr. 2010.

\bibitem{Mishali_multicoset}
------, ``Blind multi-band signal reconstruction: Compressed sensing for analog
  signals,'' \emph{IEEE Trans. on Signal Processing}, vol.~57, no.~3, pp.
  993--1009, Mar. 2009.

\bibitem{MagazineMishali}
------, ``Sub-{N}yquist sampling: Bridging theory and practice,'' \emph{IEEE
  Signal Proc. Magazine}, vol.~28, no.~6, pp. 98--124, Nov. 2011.

\bibitem{Leus}
D.~{D. Ariananda} and G.~{Leus}, ``Compressive wideband power spectrum
  estimation,'' \emph{IEEE Trans. on Signal Processing}, vol.~60, pp.
  4775--4789, Sept. 2012.

\bibitem{Leus2}
D.~{D. Ariananda}, G.~{Leus}, and Z.~Tian, ``Muti-coset sampling for power
  spectrum blind sensing,'' \emph{Int. Conf. on Digital Signal Processing
  (DSP)}, pp. 1--8, Jul. 2011.

\bibitem{Davies}
M.~{A. Lexa}, M.~{E. Davies}, J.~{S. Thompson}, and J.~{Nikolic}, ``Compressive
  power spectral density estimation,'' \emph{IEEE ICASSP}, 2011.

\bibitem{Davies2}
M.~{A. Lexa}, M.~{E. Davies}, and J.~{S. Thompson}, ``Compressive and
  noncompressive power spectral density estimation from periodic nonuniform
  samples,'' \emph{submitted for publication}, 2011.

\bibitem{Bresler}
R.~{Venkataramani} and Y.~{Bresler}, ``Perfect reconstruction formulas and
  bounds on aliasing error in sub-{N}yquist nonuniform sampling of multiband
  signals,'' \emph{IEEE Trans. Inf. Theory}, vol.~46, pp. 2173 -- 2183, Sept.
  2000.

\bibitem{ruler}
J.~{Leech}, ``On the representation of 1,2,...,n by differences,''
  \emph{Journal of the London Mathematical Society}, vol.~31, pp. 160--169,
  Apr. 1956.

\bibitem{wang}
C.~{P. Yen}, Y.~Tsai, and X.~Wang, ``Wideband spectrum sensing based on
  sub-{N}yquist sampling,'' \emph{Transactions on Signal Processing}, vol.~61,
  pp. 3028--3040, Jun. 2013.

\bibitem{radar}
O.~Bar-Ilan and Y.~{C. Eldar}, ``Sub-{N}yquist radar via {D}oppler focusing,''
  \emph{submitted to IEEE Transactions on Signal Processing}, Nov. 2012.

\bibitem{LandauCS}
H.~Landau, ``Necessary density conditions for sampling and interpolation of
  certain entire functions,'' \emph{Acta Math}, vol. 117, pp. 37--52, Jul.
  1967.

\bibitem{Papoulis}
A.~Papoulis, \emph{Probability, Random Variables, and Stochastic
  Processes}.\hskip 1em plus 0.5em minus 0.4em\relax McGraw Hill, 1991.

\bibitem{Laska2007}
J.~Laska, S.~Kirolos, M.~Duarte, J.~Laska, S.~Kirolos, and M.~Duarte, ``Theory
  and implementation of an analog-to-information converter using random
  demodulation,'' in \emph{IEEE ISCAS}, May 2007.

\bibitem{TianLeus}
G.~Leus and Z.~Tian, ``Recovering second-order statistics from compressive
  measurements,'' \emph{IEEE CAMSAP}, pp. 337--340, Dec. 2011.

\bibitem{CSBook}
Y.~{C. Eldar} and G.~Kutyniok, \emph{Compressed Sensing: Theory and
  Applications}.\hskip 1em plus 0.5em minus 0.4em\relax Cambridge University
  Press, 2012.

\bibitem{Castro}
E.~{Arias-Castro} and Y.~{C. Eldar}, ``Noise folding in compressed sensing,''
  \emph{IEEE Signal Proc. Letters}, vol.~18, no.~8, pp. 478--481, Aug. 2011.

\end{thebibliography}

\end{document}